\documentclass[11pt, journal, draftcls, onecolumn]{IEEEtran}
\usepackage[latin1]{inputenc}
\usepackage{amsmath}
\usepackage{nicefrac,xfrac}
\usepackage{amsfonts}
\usepackage{amsthm}
\usepackage{amssymb}
\usepackage{graphicx}
\usepackage{xcolor}
\usepackage{hyperref}
\usepackage[margin=0.5in]{geometry}
\usepackage[numbers,square, sort&compress]{natbib}

\newtheorem{lemma}{\textit{Lemma}}
\newtheorem{cor}{\textit{Corollary}}
\newtheorem{theorem}{\textit{Theorem}}
\newtheorem{proposition}{\textit{Proposition}}
\newenvironment{remark}[1][Remark]{\begin{trivlist}
\item[\hskip \labelsep {\bfseries #1}]}{\end{trivlist}}

\begin{document}
\title {A Unified Framework for Problems on Guessing, Source Coding and Task Partitioning}

\author{\IEEEauthorblockN{M. Ashok Kumar\IEEEauthorrefmark{1}, Albert Sunny\IEEEauthorrefmark{1}, Ashish Thakre\IEEEauthorrefmark{2}, Ashisha Kumar\IEEEauthorrefmark{2}}
		
\thanks{\IEEEauthorrefmark{1}Indian Institute of Technology Palakkad, Kerala, India
Email: \{ashokm, albert\}@iitpkd.ac.in}
\thanks{\IEEEauthorrefmark{2}Indian Institute of Technology Indore, Madhya Pradesh, India
	Email: ashishblthakre10@gmail.com, akumar@iiti.ac.in}
}

\maketitle

\begin{abstract}
We study four problems namely, Campbell's source coding problem, Arikan's guessing problem, Huieihel et al.'s memoryless guessing problem, and Bunte and Lapidoth's task partitioning problem. We observe a close relationship among these problems. In all these problems, the objective is to minimize moments of some functions of random variables, and R\'enyi entropy and Sundaresan's divergence arise as optimal solutions. This motivates us to establish a connection among these four problems. In this paper, we study a more general problem and show that R\'{e}nyi and Shannon entropies arise as its solution. We show that the problems on source coding, guessing and task partitioning are particular instances of this general optimization problem, and derive the lower bounds using this framework. We also refine some known results and present new results for mismatched version of these problems using a unified approach. We strongly feel that this generalization would, in addition to help in understanding the similarities and distinctiveness of these problems, also help to solve any new problem that falls in this framework.
\end{abstract}


\section{Introduction}
\label{sec:introduction}
Information Theory is the mathematical theory of communication
and was founded by Claude E. Shannon in his seminal paper \cite{Shannon}. The concept of entropy is most central to information theory. For a probability distribution $P = \{P(x) , x \in \mathcal{X}\}$ with $P(x) > 0$, the {\em entropy} of $P$ is defined as
\begin{equation}
\label{eqn:shannon_entropy}
H(P) := \sum \nolimits_{x \in \mathcal{X}} P(x) \log \frac{1}{P(x)}.
\end{equation}

Unless specifically stated, all the logarithms mentioned in the paper are to the base 2.

If $X$ is a random variable that follows $P$, the above quantity also refers to the entropy of $X$, and is denoted by $H(X)$. Entropy arises in various problems of information theory. In particular, Shannon showed that, in source coding, it is the expected code length (per letter) required to compress a source with alphabet set $\mathcal{X} = \{a_1,\ldots,a_M\}$ and probability distribution $P$. If the compressor does not know the true distribution $P$, but assumes a distribution $Q$, then the optimal expected code length (per letter) is $H(P) + I(P,Q)$, where
\begin{equation}
\label{eqn:Kullback}
I(P,Q) := \sum \nolimits_{x \in \mathcal {X}} P(x) \log\frac{P(x)}{Q(x)},
\end{equation}
is known as the \emph{entropy of $P$ relative to $Q$} or the \emph{Kullback-Leibler divergence}. $I(P,Q)$ can be interpreted as the penalty for not knowing the true distribution. 

In his seminal paper, Shannon also argued that entropy can be regarded as a measure of uncertainty. According to him, {\em ``if there is such a measure, say $H(p_1, p_2,\ldots, p_n)$, it is reasonable to require of it the following properties.
	\begin{enumerate}
		\item (Continuity) $H$ is continuous in $p_i$'s.
		
		\item (Monotonicity) If all the $p_i$ are equal, $p_i = 1/n$, then $H$ should be a monotonic increasing function of $n$. With equally likely events there is more choice, or uncertainty, when there are more possible events.
		
		\item (Strong additivity) If a choice be broken down into many successive choices, the original H should be the weighted sum of the individual values of H, that is, 
		\begin{align}
		&H(p_1 p_{11},\ldots,p_1 p_{1n}, \ldots, p_m p_{m1},\ldots,p_m p_{mn})  
		= H(p_1,\ldots,p_m) + \sum \nolimits_{i=1}^{m}p_i H(p_{i1}, p_{i2},\ldots, p_{in})." \label{eqn:strong_additive}
		\end{align}
		
	\end{enumerate}
} 

In terms of random variables, property 3 translates to, if $X$ and $Y$ are random variables, then $H(X,Y) = H(X) + H(Y|X)$, where $H(Y|X)$ is the conditional entropy of $Y$ given $X$. Indeed, Shannon proved that the only $H$ that satisfies the above three properties is the quantity given by Eqn.~\eqref{eqn:shannon_entropy} \cite[Th. 2]{Shannon}.

In 1961, R\'enyi \cite{Renyi} proved that, in the above, if one replaces property 3 by a weaker variant, namely
\begin{enumerate}
	\item[3'.] (Additivity)
	\begin{eqnarray}
	\label{eqn:additive}
	H(p_1 q_1,p_1 q_2,\ldots,p_1 q_n, \ldots, p_m q_1,p_m q_2,\ldots,p_m q_n) = H(p_1,\ldots,p_m) + H(q_1, \ldots, q_n),
	\end{eqnarray}
that is, $H(X,Y) = H(X) + H(Y)$ if $X$ and $Y$ are independent,
\end{enumerate}
then the quantities, now known as {\em R\'enyi entropy} of order $\alpha$, 
\begin{equation}
\label{eqn:renyi_entropy}
H_{\alpha}(P):=\frac{1}{1-\alpha}\log \sum \nolimits_{x \in \mathcal{X}} P(x)^{\alpha},
\end{equation}
where $\alpha >0$ and $\alpha \neq 1$, also satisfy properties 1, 2, and 3'. Refer Aczel and Daroczy \cite{Aczel} and the references therein for an extensive study of characterizations of information measures. By a simple application of the L'H\^ospital rule, one can show that $\lim_{\alpha\to1}H_{\alpha}(P) = H(P)$. Hence, R\'enyi entropy can be regarded as a generalization of the Shannon entropy.

In 1965, Campbell \cite{Campbell} gave an operational meaning to R\'enyi entropy. He showed that, instead of the expected code lengths, if one minimizes the cumulants of code lengths, then the optimal cumulant is the R\'enyi entropy. He also showed that the optimal cumulant can be achieved by encoding sufficiently long sequences of symbols. In 1988, Blumer and McElice \cite{BlumerMcElice} (c.f. Sundaresan \cite[Th. 8]{Sundaresan}) obtained the redundancy in cumulants in the mismatched case. Indeed, the optimal cumulant, assuming that the true underlying distribution is $Q$, when it is actually $P$, is $H_{\alpha}(P) + I_{\alpha}(P,Q)$, where 

{\small
\begin{align}
I_\alpha (P,Q) := 
\frac{\alpha}{1-\alpha}{\rm log}\Big ( \sum_ {x \in \mathcal{X}} P(x)
\Big [ & \sum_ {x' \in \mathcal{X}}
\Big ( \frac{Q(x')}{Q(x)}\Big )^\alpha \Big ]^\frac{1- \alpha}{\alpha}\Big ) - H_\alpha(P), \label{eqn:sundaresan_divergence}
\end{align}
}

\noindent
called {\em $\alpha$-entropy of $P$ relative to $Q$} or the {\em Sundaresan's} divergence \cite{Sundaresan02ISIT}. Thus, $I_{\alpha}(P,Q)$ is the penalty for not knowing the true distribution in the Campbell's coding problem. $I_{\alpha}(\cdot, \cdot)$ was also identified by Johnson and Vignat in the context of maximum R\'enyi entropy \cite{JohnsonV07J}. They called it {\em relative $\alpha$-R\'enyi entropy}. $I_{\alpha}(\cdot,\cdot)$ for $\alpha >1$ also arises in robust inference problems (see \cite{KumarS15J2} and the references therein).

In 1994, Massey \cite{Massey} studied a problem on guessing where one is interested in the expected number of guesses required to guess a random variable $X$ that assumes values from an infinite set, and found a lower bound in terms of Shannon entropy. Arikan \cite{Arikan} studied it for finite alphabet set and showed that R\'enyi entropy arises as the optimal solution for the moments of number of guesses. Later, Sundaresan \cite{Sundaresan} studied the redundancy in mismatched guessing. Indeed, he showed that the penalty in guessing according to a distribution $Q$ when the true distribution is $P$ is measured by $I_\alpha(P,Q)$.

Bunte and Lapidoth \cite{Bunte} studied the problem of partitioning of tasks and showed that the R\'enyi entropy and Sundaresan's divergence play a similar role in the optimal solution of the moments of the number of tasks performed.

Quite recently, Huieihel et al. \cite{Huieihel} studied the memoryless guessing problem, which is a variant of Arikan's guessing problem in which the guesses are i.i.d. (independent and identically distributed). They showed that R\'enyi entropy arises as the exact solution of the optimal factorial moments of the number of guesses. 


On studying the aforementioned four problems, we could see a close relationship among them. In all these problems, the objective is to minimize moments or factorial moments of random variables, and the R\'enyi entropy and Sundaresan's divergence arise as optimal solutions. This motivates us to find a connection between these four problems. Indeed, we solve all these problems by an unified approach. 


\subsection{Our Contribution}
The main contributions of this paper are as follows
\begin{itemize}
    \item Studying a more general problem and showing that R\'{e}nyi and Shannon entropies as its solution.
    \item Casting problems on source coding, guessing and task partitioning as particular instances of this general problem, and deriving the lower bounds using this framework.
    \item A unified approach to derive bounds for the mismatched version of these problems in terms of Kullback-Leibler and Sunderesan divergences.
    \item Refinement of some known results for the mismatched version of these problems. 
    \item New results on mismatched version of memoryless guessing and task partitioning.
\end{itemize}

The remainder of the paper is organized as follows. In Section~\ref{sec:probelm_statements}, we describe all the problems and state the results (both known as well as new). In Section~\ref{sec:unified_approach}, we formulate a unified approach which allows us to derive information theoretic bounds for these problems. Proofs of the stated results are derived using the unified approach in Section~\ref{sec:proofs_unified}. Finally, we summarize and conclude the paper in Section~\ref{sec:summary}.

\section{Problem Statements and Known Results}
\label{sec:probelm_statements}
In this section we discuss all the four problems in detail and present some known and new results. 
\subsection{Source Coding Problem}
A source consists of an alphabet set $\mathcal{X} = \{a_1,\ldots,a_M\}$ together with a probability distribution $P$. A {\em binary code} $C$ is a mapping from $\mathcal{X}$ to the set of finite length binary strings. Let $X$ be a random variable that assumes values from the set $\mathcal{X}$ and follows a probability distribution $P$. Let $C(X)$ be the codeword assigned to $X$. Let $L(C(X))$ be the length of codeword $C(X)$. An obvious condition for the codewords is that they are uniquely decodable. Among uniquely decodable codes, the objective is to find the one that minimizes the expected code-length. That is, 
$$\textrm{Minimize}\quad \mathbb{E}[L(C(X))]$$ 
over all uniquely decodable codes $C$. However, Kraft-McMillan proved the following.

\begin{proposition} \label{prop:KraftMcmilan}
	\cite{Cover} If $C$ a uniquely decodable code, then
	\begin{equation}
	\label{eqn:KraftMcmilan}
	\sum \nolimits_{x \in \mathcal{X}} 2^{-L(C(x))}\leq 1.
	\end{equation}
	Conversely, given a length sequence that satisfies the above inequality, there exists a uniquely decodable code $C$ with the given length sequence.
\end{proposition}

Thus, one can confine the search space of the expected code-length minimization to codes satisfying the Kraft-McMillan inequality.
\begin{theorem}
 \label{thm:source_coding_thm1}
{\rm \cite[Th.~5]{Shannon}} Let $X$ be a random variable on the source $(\mathcal{X}, P)$. Let $C$ be a uniquely decodable code. Then
\begin{equation*}
	\mathbb{E}[L(C(X))] \ge H(P).
\end{equation*}
\end{theorem}
\medskip

Shannon also proved that the optimal code-length can  be achieved by encoding long sequences of symbols. Indeed, he proved the following. Let $\mathcal{X}^n$ denote the cartesian product of $\mathcal{X}$ with itself $n$ times;  representing all possible $n$-length sequences of symbols from $\mathcal{X}$. Let $X^n := X_1,\ldots,X_n$ denote an i.i.d. sequence from $\mathcal{X}^n$, and $P_n$ denote the product distribution of $X^n$, that is, $P_n(X^n) = \prod_{i=1}^n P(X_i)$.

\begin{theorem}	\label{thm:source_coding_thm2}
{\rm \cite[Th.~5]{Shannon}} Let $X^n$ be an i.i.d. sequence from $\mathcal{X}^n$ following the distribution $P_n$. Let $C_n$ be a code such that $L(C_n(X^n)) = \lceil -\log P_n(X^n) \rceil$. Then, $C_n$ satisfies the Kraft-McMillan inequality and 
	\begin{equation*}
	\lim_{n\to\infty}\frac{L(C_n(X^n))}{n} = H(P).
	\end{equation*}
\end{theorem}
\medskip

\noindent
\textbf{Mismatch Case} 
\medskip

If the compressor does not know the true distribution $P$, but encodes according to the optimal length corresponding to some distribution $Q$, then the minimum expected code-length is bounded above by
\begin{align*}
 \sum \nolimits_{x \in \mathcal{X}} P(x) \left\lceil\log\left(\frac{1}{Q(x)}\right)\right\rceil
&\leq \sum \nolimits_{x \in \mathcal{X}} P(x)\log\left(\frac{1}{Q(x)}\right ) + 1 = \sum \nolimits_{x \in \mathcal{X}} P(x)\log\left(\frac{P(x)}{Q(x)P(x)}\right) + 1\\
&= \sum \nolimits_{x \in \mathcal{X}} P(x) \log\left(\frac{1}{P(x)}\right) + \sum \nolimits_{x \in \mathcal{X}} P(x) \log\left(\frac{P(x)}{Q(x)}\right) + 1\\ 
&= H(P) + I(P,Q) + 1.
\end{align*}

Thus, $I(P,Q)$ is the penalty for not knowing the true distribution. $I(P,Q)$ is a {\em divergence function} in the sense that $I(P,Q) \ge 0$ and $I(P,Q) = 0$ if and only if $P = Q$. A result similar to Theorem~\ref{thm:source_coding_thm2} can also be established for the mismatched case.

\subsection{Campbell Coding Problem}

In Campbell's coding problem, one minimizes the normalized cumulant of code lengths, that is,
\begin{equation*}
\text{Minimize} \quad\frac{1}{\rho}\log\mathbb{E}[2^{\rho L(C(X))}],
\end{equation*}
over all uniquely decodable codes $C$, and $\rho >0$.

Minimization of cumulants arise in the probability of buffer overflow in coding with a buffer at the encoder \cite{Humblet}. Moments of random variable also help us understand various characteristics of the underlying probability distribution. With the help of moments, one can find the mean, variance, skewness, kurtosis, etc of a distribution. The underlying distribution can also be determined by its moment sequence. Campbell \cite{Campbell} gave a lower bound for the normalized cumulants in terms of R\'enyi entropy.

\begin{theorem} \label{thm:campbell_coding_thm1}
{\rm \cite[Lemma]{Campbell}} Let $L(C(\cdot))$ be the length function of a uniquely decodable code $C$.
	Then
	\begin{equation}
	\label{eqn:Campbell_lower_bound}
	\frac{1}{\rho}\log\mathbb{E}[2^{\rho L(C(X))}]\geq H_{\alpha(\rho)}(P),
	\end{equation}
	where $\alpha(\rho) = {1}/{(1+\rho)}$ and $H_{\alpha(\rho)}(P)$ is as in \eqref{eqn:renyi_entropy}.
\end{theorem}

For the sake of convenience we suppress the dependence of $\rho$ on $\alpha$ and simply write $\alpha(\rho)$ as $\alpha$ in the remainder of this paper. Campbell also showed that, if we ignore the integer constraint of the length function, then
\begin{equation}
	\label{eqn:generalized_code_lengths}
	L(C(x))= \log  \frac{Z_{P,\alpha}}{P(x)^\alpha},
	\end{equation}
	where $Z_{P, \alpha} := \sum_{x\in \mathcal{X}} P(x)^{\alpha}$, satisfies (\ref{eqn:KraftMcmilan}) and achieves the lower bound in (\ref{eqn:Campbell_lower_bound}). Indeed, he showed that the lower bound in (\ref{eqn:Campbell_lower_bound}) can be achieved by encoding long sequences of symbols with code-lengths close to (\ref{eqn:generalized_code_lengths}).

\begin{theorem} \label{thm:campbell_coding_thm2}{\rm \cite[Theorem]{Campbell}}
If $C_n$ is a uniquely decodable code such that
\begin{equation}
\label{eqn:optimal_code_length_campbell}
	L(C(x^n))= \Big\lceil\log  \frac{Z_{P,\alpha,n}}{P(x^n)^\alpha}\Big\rceil,
	\end{equation}
	where $Z_{P, \alpha, n} := \sum_{x^n \in \mathcal{X}^n} P(x^n)^{\alpha}$, then
\begin{equation}
\label{eqn:limit_result_campbell}
\lim_{n\to\infty}\frac{1}{n\rho}\log \mathbb{E}[2^{\rho L(C_n(X^n))}] = H_{\alpha}(P).
\end{equation}
\end{theorem}

\bigskip
\noindent
\textbf{Mismatch Case}
\medskip

Blumer and McEliece \cite{BlumerMcElice} and Sundaresan \cite{Sundaresan} studied the redundancy in the mismatched case of the Campbell's problem. Sundaresan showed that the difference in the normalized cumulant from the minimum when encoding according to an arbitrary uniquely decodable code is measured by the $I_{\alpha}$-divergence upto a factor of 1. That can be generalized to the following.

\begin{theorem} \label{thm:campbell_coding_mismatch}
Let $\rho\in (-1,0)\cup (0,\infty)$. Let $L: \mathcal{X} \to \mathbb{Z}_{+}$ be an arbitrary length function that satisfies (\ref{eqn:KraftMcmilan}). Let
\begin{equation}
\label{eqn:difference_cumlant}
    R_c(P,L,\rho) := \frac{1}{\rho}\log E[2^{\rho L(X)}] -\min_K \frac{1}{\rho}\log E[2^{\rho K(X)}],
\end{equation}
where the expectation is with respect to $P$ and the minimum is over all length functions $K$ satisfying (\ref{eqn:KraftMcmilan}). Then, there exists a probability distribution $Q_L$ such that
\begin{equation}
    \label{eqn:bounds_difference_cumulant}
    I_{\alpha}(P,Q_L) - \log \eta - 1 \le R_c(P,L,\rho) \le I_{\alpha}(P,Q_L) - \log \eta,
\end{equation}
where $\eta  = \sum_x 2^{-L(x)}$. 
\end{theorem}
\medskip

It should be noted that the lower bound in \eqref{eqn:bounds_difference_cumulant} is arbitrarily large for sufficiently small $\eta$. For example, for a source with two symbols, say $x$ and $y$, with code lengths $L(x) = L(y) = 100$, the left side of \eqref{eqn:bounds_difference_cumulant} becomes $I_{\alpha}(P,Q_L) + 98$. However, if $\nicefrac{1}{2} \leq \eta \leq 1$, then \eqref{eqn:bounds_difference_cumulant} becomes
$$
|R_c(P,L,\rho) - I_{\alpha}(P,Q_L)| \le 1,
$$
which is {\rm \cite[Th.~8]{Sundaresan}}. $I_{\alpha}(P,Q_L)$ is, in a sense, the penalty when $Q_L$ is not matched with the true distribution $P$. $I_\alpha(\cdot, \cdot)$ is also a {\em divergence function} in the sense that
\begin{enumerate}
	\item $I_\alpha(P,Q)\geq0.$
	\item $I_\alpha(P,Q)=0$ if and only if $P=Q$.
\end{enumerate}
Moreover, as $\alpha\to 1$, $I_\alpha(P,Q)\to I(P,Q)$ \cite{Sundaresan}. In view of this, a result analogous to Theorem \ref{thm:campbell_coding_mismatch} also holds for the Shannon source coding problem.

\subsection{Arikan's Guessing Problem}

Let $\mathcal{X}$ be a set of $M$ objects. Bob thinks of an object $X$ (a random variable) from $\mathcal{X}$ according to a probability distribution $P$. Alice guesses it by asking questions of the form ``is $X=x$?". The objective is to minimize the average number of guesses required for Alice to correctly guess $X$. Suppose Alice uses a guessing strategy $G$ for asking questions to Bob. By a guessing strategy, we mean a one-one map $G:\mathcal{X}\to\{1,\ldots,M\}$, where $G(x)$ is to be interpreted as the number of questions required by Alice to guess $x$ correctly.
Arikan studied the $\rho^{\textrm{th}}$ moment of number of guesses and found upper and lower bounds in terms of R\'enyi entropy.

\begin{theorem} \label{thm:massey_guessing_thm1}	
 Let $\rho \in (-1,0) \cup (0,\infty)$. Let $X$ be a random guess according to the distribution $P$, and let $G $ be any guessing function. Then 
$$\frac{1}{\rho}\log\mathbb{E}\big[G(X)^\rho\big]\geq H_\alpha(P)-\log h_M \footnote{Arikan upper bounded $h_M$ by $1+\log M$.},$$
where $h_M = 1 + 1/2 + \cdots + 1/M$.
\end{theorem}

Arikan also showed that, if one follows the optimal guessing function $G^*$, that is, guessing according to the decreasing order of $P$-probabilities, it does not require more than the R\'enyi entropy.

\begin{theorem} \label{thm:massey_guessing_thm2}
	Let $\rho \in (-1,0) \cup (0,\infty)$. If $G^*$ is an optimal guessing function, then
	\begin{equation*}
	 \frac{1}{\rho}\log\mathbb{E}\big[G^*(X)^\rho\big]\leq H_\alpha (P).
	\end{equation*}
\end{theorem}
\medskip

Arikan also proved that R\'enyi entropy can be achieved by guessing long sequences of letters in an i.i.d. fashion.
\begin{theorem} \label{thm:massey_guessing_thm3}
Let $\rho \in (-1,0) \cup (0,\infty)$. Let ${X_1,X_2,\ldots,X_n}$ be a sequence of i.i.d. guesses. Let $G_n^*(X_1,\ldots,X_n)$ be an optimal guessing function. Then, 
$$  \lim_{n \to \infty}\frac{1}{n\rho} \log \mathbb{E}[G_n^*(X_1,X_2,\ldots,X_n)^\rho] = H_\alpha(X_1).
$$	
\end{theorem}

\begin{remark}
For the case when $\rho >0$, Theorems~\ref{thm:massey_guessing_thm1}, \ref{thm:massey_guessing_thm2} and \ref{thm:massey_guessing_thm3} were stated and proved in \cite{Arikan}.
\end{remark}

\bigskip
\noindent
\textbf{Mismatch Case}
\medskip

Suppose Alice does not know the true underlying probability distribution $P$, and she guesses according to some guessing function $G$. Sundaresan proved that the penalty in this mismatched guessing is again measured by the $I_{\alpha}$-divergence \cite[Th.~6]{Sundaresan}. This can be refined to the following theorem.

\begin{theorem} \label{thm:guessing_mismatch_thm1}
 Let $\rho \in (-1,0) \cup (0,\infty)$. Let $X$ be a random variable that assumes values from $\mathcal{X}$ and follows the distribution $P$. Let $G$ be an arbitrary guessing function. Then, there exists a probability distribution $Q_G$ on $\mathcal{X}$ such that
$$\frac{1}{\rho}\log\mathbb{E}[G(X)^{\rho}] = H_\alpha (P) + I_\alpha (P, Q_G)- \log h_M,$$
where the expectation is with respect to $P$.
\end{theorem}

\begin{theorem} \label{thm:guessing_mismatch_thm2}
{\rm \cite[Prop.~1]{Sundaresan}} Let $\rho \in (-1,0) \cup (0,\infty)$. Let $X$ be a random variable that assumes values from $\mathcal{X}$ and follows the distribution $P$. Let $Q$ be another distribution. Let $G_Q$ be the guessing strategy that guesses according to the decreasing order of $Q$-probabilities. Then
	\begin{equation*}
	\label{gum}
	\frac{1}{\rho}\log \mathbb{E}\big[G_Q(X)^\rho \big]	\leq H_\alpha(P) + I_\alpha(P,Q),
	\end{equation*}
	where expectation is with respect to $P$.
\end{theorem}

The above two theorems can be combined to state the following, as in the Campbell's coding problem.

\begin{theorem} \label{thm:guessing_mismatch_thm3}
 Let $\rho \in (-1,0) \cup (0,\infty)$. Let $X$ be a random variable that assumes values from $\mathcal{X}$ and follows the distribution $P$. Let $G$ be an arbitrary guessing function and $G_P$ be the optimal guessing function for $P$. Let
\begin{equation*}
\label{eqn:difference_guessing_moment}
R_{g}(P,G,\rho) := \frac{1}{\rho}\log \mathbb{E}\big[G(X)^\rho \big] - \frac{1}{\rho}\log \mathbb{E}\big[G_P(X)^\rho \big],
\end{equation*}
where the expectation is with respect to $P$. Then, there exists a probability distribution $Q_G$ such that
$$I_\alpha (P, Q_G) - \log h_M \leq R_{g}(P,G,\rho) \leq I_\alpha (P, Q_G)$$
\end{theorem}

\begin{remark}
For the case when $\rho >0$, the above theorem was stated and proved in {\rm \cite[Th.~6]{Sundaresan}}.
\end{remark}

\subsection{Memoryless Guessing}
\label{subsec:memoryless}
In Arikan's guessing problem, the guesser Alice does not ask a question that was asked before. We now present a memoryless guessing where Alice does not keep track of her previous guesses; every-time she comes up with a guess independent of her previous guesses. Let $\hat{X}_1, \hat{X}_2, \ldots$ be a sequence of independent guesses according to a distribution $\hat{P}$. Let $P$ be the underlying distribution with respect to which Bob sets the guesses. Huieihel et al. \cite{Huieihel} showed that, surprisingly, the optimal one for Alice, even if she knows the true underlying distribution, is not $P$, but $P^{(\alpha)}$, the $\alpha$-scaled measure of $P$, defined by $P^{(\alpha)}(x) = {P(x)^\alpha}/Z_{P,\alpha}$. The guessing function in this problem is defined as,
\begin{equation*}
G(X,\hat{X}_{1}^{\infty}) := \inf\{K \geq 1 : \hat{X}_{K} = X\},
\end{equation*}
that is, the number of guesses until a success.

Unlike in Massey's guessing problem, Huieihel et al. \cite{Huieihel} minimized what are called {\em factorial moments}, defined for $\rho \in \mathbb{Z}_{+}$ as,
\begin{equation*}
 V_\rho (X, \hat{X}_{1}^{ \infty })=\frac{1}{\rho !} \prod \nolimits_{l=0}^{\rho -1}
\big (G(X,\hat{X}_{1}^{\infty}) + l\big ).
\end{equation*}
Huieihel et al. \cite{Huieihel} studied the following problem.
\begin{equation*}
\text{Minimize}\quad\mathbb{E}\big[V_\rho (X, \hat{X}_{1}^{ \infty })\big],
\end{equation*}
over all $\hat{P}\in \mathcal{P}$, where $\mathcal{P}$ is the probability simplex, that is, $\mathcal{P} = \{(P(x))_{x \in \mathcal{X}} : P(x)\geq 0,\sum_x P(x)=1\}$. The optimal value, if attained, is denoted by $\mathbb{E}\big[V^*_\rho (X, \hat{X}_{1}^{ \infty })\big]$. Huieihel et al. \cite{Huieihel} proved the following.
\medskip
\noindent

\begin{theorem} \label{thm:memoryless_guessing_thm1}
{\rm \cite[Th.~1]{Huieihel}} For any integer $\rho>0$, we have
$$ \frac{1}{\rho}{\rm log}\mathbb{E}\big[V^*_\rho (X, \hat{X}_{1}^{ \infty })\big]
 = H_{\alpha}(P)$$
and is attained by
$$\hat{P}^*(x)=\frac{P(x)^\alpha}{Z_{P,\alpha}}.$$
\end{theorem}
\medskip
\noindent

For sequence of random guesses $\hat{X}^n$, the above theorem can be stated in the following way. Let $\hat{X}^n = (\hat{X}_1,\ldots,\hat{X}_n$) be a sequence of i.i.d. guesses from $\mathcal{X}^n$ with distribution $\hat{P}_n$ --- the $n$-fold product distribution of $P$ on $\mathcal{X}^n$. Assume that $P_n$ is the true underlying distribution. Then
$$\lim_{n\to\infty}\frac{1}{n\rho}{\rm log}\mathbb{E}\big[V^*_\rho (X^n, \hat{(X^n)}_1^{\infty})\big] = H_{\alpha}(P),$$
where the expectation is with respect to $P_n$.

It is now easy to obtain the result in the mismatched case.

\begin{theorem}
\label{thm:memoryless_mismatch}
If the true underlying probability distribution is $P$ but Alice thinks that it is $Q$ and guesses according to its optimal one, namely $\widehat{Q}^{*}(x) = Q(x)^{\alpha} / Z_{Q,\alpha}$, then
$$\frac{1}{\rho} \log \mathbb{E}[V^{*}_{\rho}(X,\widehat{X}^{\infty}_1)] = H_{\alpha}(P) + I_{\alpha}(P,Q),$$
where the expectation is with respect to $P$.
\end{theorem}

\subsection{Tasks Partitioning Problem}
Let $\mathcal{X}$ be a finite set of tasks. A task $X$ is randomly drawn from $\mathcal{X}$ according to a probability distribution $P$ and is performed. The distribution $P$ may correspond to the frequency of occurrences of the tasks. Suppose there are only $N$ keys to which these tasks are to be associated. Typically, $N < |\mathcal{X}|$. Due to the limited availability of keys, more than one task may be associated with a single key. When a task needs to be performed, the key associated with the task is pressed and all the tasks associated with the key will be performed. The objective in this problem is to minimize the number of redundant tasks performed. Usual coding techniques suggest us to assign tasks with higher probability to individual keys and leave the lower probability tasks unassigned to any keys. But for an individual all the tasks can be equally important. Just that some tasks may have more frequency of occurrences while others have lesser. If $N \geq |\mathcal{X}|$, then one can perform tasks without any redundancy. However, Bunte and Lapidoth \cite{Bunte} show that, even when $N<|\mathcal{X}|$, one can accomplish the tasks with much less redundancy on average, provided the underlying probability distribution is different from the uniform distribution. When $N <  |\mathcal{X}|$ the only way is, to partition the set of tasks and assign keys in an intelligent way keeping the probability distribution in mind, so as to minimize the average number of redundant tasks performed. 

Let $\mathcal{A} = \{\mathcal{A}_1,\mathcal{A}_2,\ldots,\mathcal{A}_N\}$ be a partition of $\mathcal{X}$ that corresponds to the assignment of tasks to keys. Let $A(x)$ be the cardinality of the subset containing $x$ in the partition. We shall call $A$ the {\em partition function} associated with $\mathcal{A}$. Bunte and Lapidoth \cite[Th.~I.1]{Bunte} proved the following theorem for the case $\rho >0$.

\begin{theorem} \label{thm:task_partition_thm1}
Let $\rho \in (-1,0) \cup (0, \infty)$. Let $X$ be a random task from $\mathcal{X}$ following distribution $P$. Then the following hold.
	\begin{enumerate}
	\item[(a)] For any partition of $\mathcal{X}$ of size $N$ with partition function $A$, we have
		$$\frac{1}{\rho}\log \mathbb{E}[A(X)^\rho]\geq H_\alpha(P)-\log N.$$
		
	\item[(b)]
	\begin{enumerate}
	    
		\item[(i)] If $N>\log |\mathcal{X}|+2$ and $\rho > 0$, then there exists a partition of $\mathcal{X}$ of size at most $N$ with partition function $A$ such that
		$$1 \leq \mathbb{E}[A(X)^\rho] \leq 1+2^{\rho(H_\alpha(P)-\log\tilde{N})},$$
		where $\tilde{N} = \frac{N-\log |\mathcal{X}| - 2}{4}$ .
		
	    \item[(ii)] If $\rho < 0$, then there exists a partition of $\mathcal{X}$ of size at most $N$ with partition function $A$ such that
		$$\frac{1}{2} 2^{\rho H_\alpha(P)} \leq \mathbb{E}[A(X)^\rho] \leq 1.$$
	\end{enumerate}
	\end{enumerate}
\end{theorem}

Bunte and Lapidoth also found conditions under which the redundancy can be made negligible on average while performing a long sequence of tasks.

\begin{theorem} \label{thm:task_partition_thm2}
{\rm \cite[Th.~I.2]{Bunte}}
	Let $X_1,X_2,\ldots,X_n$ be a sequence of i.i.d. tasks following distribution $P$. Then the following hold.
	\begin{enumerate}
		\item[(i)] If $\log N> H_{\alpha}(P)$, there exists partition $\mathcal{A}_n$ of $\mathcal{X}^n$ of size at most $N^n$ with associated partition function $A_n$  such that
		$$\lim_{n\to \infty} \mathbb{E}[A_n(X^n)^{\rho}] = 1.$$
		
		\item[(ii)] If $\log N<H_{\alpha}(P)$, then for partitions $\mathcal{A}_n$ of $\mathcal{X}^n$ of size at most $N^n$ with associated partition function $A_n$,
		$$\lim_{n\to\infty} \mathbb{E}[A_n(X^n)^{\rho}] = \infty.$$
	\end{enumerate}
\end{theorem}

\noindent
\textbf{Mismatch Case}

Consider a scenario where one does not know the true  underlying  probability  distribution $P$, but arbitrarily partitions $\mathcal{X}$. Then, the penalty due to such a partition can be measured by the $I_{\alpha}$-divergence as stated in the following theorem.

\begin{theorem} \label{thm:task_partition_mismatch_lb}
Let $\rho \in (-1,0) \cup (0,\infty)$. Let $\mathcal{A}$ be a partition of $\mathcal{X}$ of size $N$ with partition function $A$. Then, there exists a probability distribution $Q_{A}$ on $\mathcal{X}$ such that 
$$\frac{1}{\rho}\log \mathbb{E}[A(X)^\rho] = H_\alpha(P) + I_{\alpha}(P,Q_A) -\log N.$$
\end{theorem}
\medskip

\noindent
A converse result is the following.

\begin{theorem} \label{thm:task_partition_mismatch_ub}{\rm \cite[Sec. IV]{Bunte}}
Let $\rho > 0$. Let $X$ be a random task from $\mathcal{X}$ following distribution
$P$. Let $Q$ be  another  distribution on $\mathcal{X}$. If $N > \log |\mathcal{X}| +2$, then there exists partition $\mathcal{A}_Q$ (with an associated partition function $A_Q$) of $\mathcal{X}$ of size at most $N$ such that
$$ \mathbb{E}[A_Q(X)^\rho] \leq 1+2^{\rho(H_\alpha(P) + I_\alpha(P,Q)-\log\tilde{N})}, $$
where $\tilde{N} = \frac{N-\log |\mathcal{X}| - 2}{4}$ and the expectation is with respect to $P$.
\end{theorem}


\section{A Unified Approach} \label{sec:unified_approach}


Our objective in this section is to formulate a more general problem to which the problems studied in the previous section are particular cases.

\begin{theorem}
\label{thm:lower_bound_unified}
Let $\rho\in (-1,0)\cup (0,\infty)$. Let $\psi: \mathcal{X} \to [0,\infty)$ such that $\sum_{x \in \mathcal{X}} \psi(x) \leq b$ for some $b>0$. Then 
\begin{equation}
\label{eqn:lower_bound_on_cumulant}
    \frac{1}{\rho} \log \mathbb{E}[\psi(X)^{-\rho}] \geq  H_{\alpha}(P) - \log b.
\end{equation}
The lower bound is achieved when 
\begin{equation}
    \label{eqn:equality_attain_unified}
    \psi(x) =  {b \cdot P(x)^\alpha}/{Z_{P,\alpha}} \quad \text{for } x \in \mathcal{X}.
\end{equation}
\end{theorem}

\noindent
To prove the above theorem, consider the following optimization problem
\begin{align*}
& \hspace{7mm}P_1 : \min_{\phi} \quad \textrm{sgn}(\rho) \cdot \mathbb{E}[\phi(X)^{-\rho}]\\
& \hspace{-10mm} \textrm{subject to:} \\
&  \sum \nolimits_{x \in \mathcal{X}} {\phi(x)} \leq b  \quad \textrm{and} \quad \phi(x) \geq 0 \quad \forall x \in \mathcal{X},
\end{align*}
where $\text{sgn}(\rho)$ is the signum function of $\rho$. It is easy to check that objective function of problem $P_1$ is convex for all $\rho \in (-1,0) \cup (0, \infty)$. Also, since the constraint set is compact, the minimum is indeed attained.

\begin{lemma}
\label{lem:Kraft_type_equality}
	Let $\phi^{*} = \{ \phi^{*}(x), x \in \mathcal{X}\}$ be an optimal solution of problem $P_1$. Then, $\sum_{x \in \mathcal{X}} {\phi^{*}(x)} = b$.
\end{lemma}
\begin{proof}
	Suppose $\sum_{x \in \mathcal{X}} {\phi^{*}(x)} < b$. Let $\phi^{'}(x) = \eta \phi^{*}(x)$ for $x\in\mathcal{X}$, where $\eta = {b}/{\sum_{x \in \mathcal{X}} {\phi^{*}(x)}}  > 1$. It is easy to see that $\phi^{'}$ satisfies the constraint of problem $P_1$. Further, we have
	\begin{align*}
	\textrm{sgn}(\rho) \mathbb{E}[\phi^{'}(X)^{-\rho}] = \textrm{sgn}(\rho) \eta^{-\rho} \mathbb{E}[\phi^{*}(X)^{-\rho}] \overset{(a)}{<}  \textrm{sgn}(\rho) \mathbb{E}[\phi^{*}(X)^{-\rho}],
	\end{align*}
where inequality~(a) follows because $\textrm{sgn}(\rho) \eta^{-\rho} < \textrm{sgn}(\rho)$. Hence $\phi^{*}$ cannot be optimal --- a contradiction. Thus, we must have $\sum_{x \in \mathcal{X}} {\phi^{*}(x)} = b$.
\end{proof}

\medskip
\noindent
\textbf{\textit{Proof of Theorem~\ref{thm:lower_bound_unified}}}: Due to Lemma~\ref{lem:Kraft_type_equality}, the optimal value of problem $P_1$ is same as the following optimization problem
\begin{align*}
& \hspace{7mm}P_2 : \min_{\phi} \quad \textrm{sgn}(\rho) \cdot \mathbb{E}[\phi(X)^{-\rho}]\\
& \hspace{-10mm} \textrm{subject to:} \\
&  \sum \nolimits_{x \in \mathcal{X}} {\phi(x)} = b  \quad \textrm{and} \quad \phi(x) \geq 0 \quad \forall x \in \mathcal{X}
\end{align*}
Problem $P_2$ satisfies \emph{Slater's condition}. As a consequence of this, there is no duality gap, and we can solve problem $P_2$ by solving its dual problem. The Lagrangian of problem $P_2$ is given by
\begin{align*}
L(\boldsymbol{\phi}, \lambda, \mu) = \textrm{sgn}(\rho) \sum \nolimits_{x \in \mathcal{X}}  P(x) \phi(x)^{-\rho} + \lambda  \sum \nolimits_{x \in \mathcal{X}} {\phi(x)} - \sum \nolimits_{x \in \mathcal{X}} \mu(x) \phi(x),
\end{align*}
where $\boldsymbol{\mu} = \{\mu(x),x \in \mathcal{X}\}$ is a set of  non-negative scalars. The variables $\lambda$ and $\boldsymbol{\mu}$ are known as Lagrangian multipliers. These are also known as \emph{dual variables}. An application of Karush-Kuhn-Tucker conditions \cite[Prop.~3.3.1]{2003xxNLP_Ber} tell us that, if $\phi^{*}$ is an optimal solution of problem $P_2$, then there should exist dual variables $\lambda^{*}$ and $\boldsymbol{\mu}^{*}$ such that the following hold.
\begin{align}
 -|\rho|  P(x) \phi^{*}(x)^{-(\rho+1)} + \lambda^{*} -\mu^{*}(x) &= 0 \quad \forall x \in \mathcal{X}, \label{eqn:kkt1} \\
\sum \nolimits_{x \in \mathcal{X}} {\phi^{*}(x)} &= b, \label{eqn:kkt2} \\
\phi^{*}(x) &\geq 0 \quad \forall x \in \mathcal{X}, \label{eqn:kkt_primal_ineq} \\
\mu^{*}(x) &\geq 0 \quad \forall x \in \mathcal{X}, \label{eqn:kkt3} \\
\text{and}\quad \mu^{*}(x) \phi^{*}(x) &= 0 \quad \forall x \in \mathcal{X}. \label{eqn:kkt4} 
\end{align}

\noindent
Multiplying \eqref{eqn:kkt1} by $\phi^{*}(x)$ and summing over all $x \in \mathcal{X}$ we get 
$$\lambda^{*} = \frac{|\rho|}{b}  \sum \nolimits_{x \in \mathcal{X}} P(x) \phi^{*}(x)^{-\rho}.$$

\noindent
Substituting this in \eqref{eqn:kkt1}, we obtain 
\begin{align}
\frac{|\rho|}{b}  \sum_{x \in \mathcal{X}} P(x) \phi^{*}(x)^{-\rho} = |\rho|  P(x) \phi^{*}(x)^{-(\rho+1)} + \mu^{*}(x) \quad \forall x \in \mathcal{X}. \label{eqn:kkt5}
\end{align}

\noindent
Since $\mu^{*}(x) \geq 0$, we have
\begin{equation}
\label{eqn:kkt6}
    \sum \nolimits_{x \in \mathcal{X}} P(x) \phi^{*}(x)^{-\rho} \geq  b P(x) \phi^{*}(x)^{-(\rho+1)}  \quad \forall x \in \mathcal{X}.
\end{equation}
Suppose $\phi^{*}(\hat{x}) = 0$ for some $\hat{x} \in \mathcal{X}$. Then,
\begin{itemize}
    \item for $\rho \in (0,\infty)$, \eqref{eqn:kkt6} says that the optimal value is infinite. This is not possible because we can always find a feasible point for problem $P_2$ where the objective function evaluates to a finite quantity. 
    \item for $\rho \in (-1,0)$, the left side of \eqref{eqn:kkt6} is always finite, whereas the right side evaluates to $\infty$ --- a contradiction.
\end{itemize}

Thus, we should have $\phi^{*}(x) > 0 \, \forall x \in \mathcal{X}$. This, in turn, implies that $\mu^{*}(x) = 0 \, \forall x \in \mathcal{X}$ (due to Eqn.~\eqref{eqn:kkt4}). Plugging $\mu^{*}(x) = 0$ in \eqref{eqn:kkt5}, we get 
$$\phi^{*}(x) = \kappa P(x)^{\alpha}  \quad \forall x \in \mathcal{X},$$
where $\kappa$ is the constant of proportionality. Summing over all $x \in \mathcal{X}$ and using the fact that $\sum_{x \in \mathcal{X}} {\phi^{*}(x)} = b$, we get  $\kappa = b/Z_{P,\alpha}$. Thus,  the optimal solution ${\phi}^{*}$ of problem $P_2$ has the following form
$$ \phi^{*}(x) =  \frac{b \cdot P(x)^\alpha}{Z_{P,\alpha}}$$
and 
\begin{align*}
 \mathbb{E}[\phi^{*}(X)^{-\rho}] =  b^{-\rho} \sum_{x \in \mathcal{X}} P(x) \cdot \left(\frac{Z_{P,\alpha}}{P(x)^\alpha} \right)^{\rho}=  b^{-\rho} \left( \sum_{x\in\mathcal{X}} P(x)^\alpha \right)^{\frac{1-\alpha}{\alpha}} \cdot\sum_{x \in \mathcal{X}} P(x)^\alpha = b^{-\rho} \left( \sum_{x \in \mathcal{X}} P(x)^\alpha \right)^{\frac{1}{\alpha}}.
\end{align*}
Now, if $\psi: \mathcal{X} \to [0, \infty)$ such that $\sum_{x \in \mathcal{X}} \psi(x) \leq b$, then we have
\begin{align*}
\textrm{sgn}(\rho) \cdot \mathbb{E}[\psi(X)^{-\rho}] \geq \textrm{sgn}(\rho) \cdot \mathbb{E}[\phi^{*}(X)^{-\rho}] = \textrm{sgn}(\rho) \cdot b^{- \rho} \left( \sum_{x \in \mathcal{X}} P(x)^{\alpha} \right)^{1/\alpha}.
\end{align*}

Taking log on both sides of the above inequality and multiplying throughout by $1/\rho$, we get (\ref{eqn:lower_bound_on_cumulant}).
\hfill $\qed$

\begin{theorem} \label{thm:lower_bound_logfunc}
Let $\psi: \mathcal{X} \to [0,\infty)$ such that $\sum_{x \in \mathcal{X}} \psi(x) \leq b$, then 
\begin{equation}
\label{eqn:lower_bd_exp_length}
    \mathbb{E}\left[\log \frac{1}{\psi(X)} \right] \geq  H(P) - \log b,
\end{equation}
where $H(P)$ is the Shannon entropy. The lower bound is achieved when $\psi(x)=b\cdot P(x) \,\, \forall x \in \mathcal{X}$.
\end{theorem}
\begin{proof}
It is easy to see that Lemma \ref{lem:Kraft_type_equality} also holds if the objective function is $-\mathbb{E}[\log \psi(X)]$. Hence proceeding as in Theorem \ref{thm:lower_bound_unified}, with $\textrm{sgn}(\rho) \mathbb{E}[\psi(X)^{-\rho}]$ replaced by $-\mathbb{E}[\log \psi(X)]$, we get the optimal solution $\psi^*(x) = b\cdot P(x)$. This implies (\ref{eqn:lower_bd_exp_length}).
\end{proof}


\begin{remark}
    In fact, for any convex and continuously differentiable function $f : \mathbb{R}^{+} \to \mathbb{R}$, and $\psi : \mathcal{X} \to [0, \infty)$ such that $\sum_{x \in \mathcal{X}} \psi(x) \leq b$, we have 
    \begin{align}
    \mathbb{E}[f(\psi(X))] \geq  \mathbb{E}\left[ f \left(f^{'-1} \left( \frac{\lambda^*}{p(X)} \right) \right) \right], \label{eq:general_lower_bound}
    \end{align}
    where $f^{'}$ is the derivative of $f$, $f^{'-1}$ is the inverse of $f^{'}$, and $\lambda^*$ is the solution of the equation 
    $$\mathbb{E}\left[\frac{1}{p(X)} f^{'-1} \left( \frac{\lambda^*}{p(X)} \right)\right] = b.$$
    It is interesting to notice that, while the lower bound in Inequality~\eqref{eq:general_lower_bound} is dependant on function $f$, it is independent of function $\psi$. The lower bound in \eqref{eq:general_lower_bound} is achieved when $\psi(x) = f^{'-1} \left( \frac{\lambda^*}{p(X)} \right)$. Choosing $f(y)$ as $\textrm{sgn}(\rho) \cdot (y)^{-\rho}$ and $-\log(y)$ gives us the lower bound results on R\'enyi  and Shannon entropy, respectively. The quantity on the right side of (\ref{eq:general_lower_bound}) is closely related to the $\phi$-entropies of Teboulle and Vajda \cite{TeboulleVajdaTIT93}.
\end{remark}

We now extend Theorems~\ref{thm:lower_bound_unified} and \ref{thm:lower_bound_logfunc} to sequences of random variables. Let $\mathcal{X}^n$ be the set of all sequences of length $n$ of elements of $\mathcal{X}$, and $P_n$ be the $n$-fold product distribution of $P$ on $\mathcal{X}^n$, that is, for $x^n := x_1,\dots, x_n\in \mathcal{X}^n$, $P_n(x^n) = \prod_{i=1}^n P(x_i)$.
\begin{cor}
\label{cor:sequence_bound_lower}
	Given any $n \geq 1$, if $\psi_n: \mathcal{X}^{n} \to [0,\infty)$ is such that 
	$$\sum_{x^n \in \mathcal{X}^n} \psi_n(x^n) \leq b_n$$
	for some $b_n >0$, then
	\begin{enumerate}
	    \item For $\rho \in (-1,0) \cup (0, \infty)$, 
	    $$ \liminf_{n \to \infty} \frac{1}{n\rho} \log \mathbb{E}[\psi(X^n)^{-\rho}] \geq H_{\alpha}(P)  - \limsup_{n \to \infty}\frac{\log b_n}{n}.$$
	    \item 
	    $$ \liminf_{n \to \infty} \frac{1}{n}  \mathbb{E}\left[\log \frac{1}{\psi(X^n)} \right] \geq H(P) - \limsup_{n \to \infty}\frac{\log b_n}{n},$$
	\end{enumerate}
	where the expectations are with respect to $P_n$.
\end{cor}
\begin{proof}
It is easy to see that $H_{\alpha}(P_n) = n H_{\alpha}(P)$ and $H(P_n) = n H(P)$. Hence, applying Theorems~\ref{thm:lower_bound_unified} and \ref{thm:lower_bound_logfunc}, dividing throughout by $n$ and taking $\liminf n \to \infty$, the results follow.
\end{proof}

\subsection{A General Framework for Mismatched Cases} \label{subssec:mismatch_case}
In this sub-section, we establish a unified approach for the case when there is a mismatch between the assumed and underlying distributions.

\begin{theorem}
\label{thm:bound1_unified_mismatch}
	Let $\rho \in (-1,0) \cup (0,\infty)$ and $Q$ be a probability distribution on $\mathcal{X}$. Suppose $\psi: \mathcal{X} \to [0,\infty)$ is such that 
	$$ \psi(x)^{-1} \leq \left(b + \left(\frac{c\cdot  Z_{Q,\alpha}}{Q(x)^{\alpha}} \right)^{|\rho|} \right)^{1/|\rho|}$$
	for some $c > 0$, and $b \geq 0$, then
	\begin{enumerate}
	     \item $  \mathbb{E}[\psi(X)^{-\rho}] \leq 
	      b +   2^{\rho(H_{\alpha}(P) + I_{\alpha}(P,Q) + \log c)}$\qquad if $\rho > 0$,
	      
	      \item $  \mathbb{E}[\psi(X)^{-\rho}] \geq
	      2^{\rho(H_{\alpha}(P) + I_{\alpha}(P,Q) + \log c - \rho^{-1} \log(1 + b c^{\rho}))}$\qquad if $\rho < 0$,
	\end{enumerate}
	where the expectation is with respect to $P$.
\end{theorem}
\begin{proof}

\noindent
\textbf{Proof of upper bound for the case $\boldsymbol{\rho > 0}$:} We have
$$\psi(x)^{-\rho} \leq b + \left(\frac{c\cdot Z_{Q,\alpha}}{Q(x)^{\alpha}} \right)^{\rho}.$$
Hence 
\begin{align*}
\mathbb{E}[\psi(X)^{-\rho}] = \sum \nolimits_{x \in \mathcal{X}} P(x) \psi(x)^{-\rho}  \leq b +  c^{\rho} Z_{Q,\alpha}^{\rho} \sum \nolimits_{x \in \mathcal{X}} P(x) Q(x)^{-\alpha\rho}.
\end{align*}
\noindent
But
\begin{equation}
    \label{eqn:Renyi_plus_Sundar}
    Z_{Q,\alpha}^{\rho}\cdot \sum \nolimits_{x \in \mathcal{X}} P(x) Q(x)^{\alpha-1} =  2^{\rho(H_{\alpha}(P) + I_{\alpha}(P,Q))}.
\end{equation}
Hence (1) follows.
\medskip

\noindent
\textbf{Proof of lower bound for the case $\boldsymbol{\rho < 0}$:}
In this case, we have
\begin{align*}
\psi(x)^{-\rho} &\geq \left(b + \left(\frac{c\cdot Z_{Q,\alpha}}{Q(x)^{\alpha}} \right)^{-\rho} \right)^{-1} \overset{(a)}{\geq} \left((b + c^{-\rho})\left(\frac{Z_{Q,\alpha}}{Q(x)^{\alpha}} \right)^{-\rho} \right)^{-1} \\
&= (b + c^{-\rho})^{-1} \left(\frac{Z_{Q,\alpha}}{Q(x)^{\alpha}} \right)^{\rho} = \left(1 + b c^{\rho} \right)^{-1} \left(\frac{c\cdot Z_{Q,\alpha}}{Q(x)^{\alpha}} \right)^{\rho},
\end{align*}
where inequality~(a) follows because $\frac{Z_{Q,\alpha}}{Q(x)^{\alpha}} \geq 1$ and $-\rho > 0$. Thus, we have
\begin{align*}
\mathbb{E}[\psi(X)^{-\rho}] = \sum \nolimits_{x \in \mathcal{X}} P(x) \psi(x)^{-\rho}  \geq \left(1 + b c^{\rho} \right)^{-1}  c^{\rho} Z_{Q,\alpha}^{\rho} \sum \nolimits_{x \in \mathcal{X}} P(x) Q(x)^{-\alpha\rho} 
\end{align*}
\noindent
Hence, in view of \eqref{eqn:Renyi_plus_Sundar}, (2) follows.
\end{proof}

\begin{theorem}
\label{thm:bound2_unified_mismatch}
	Let $\rho \in (-1,0) \cup (0,\infty)$ and $Q$ be a probability distribution on $\mathcal{X}$. If $\psi: \mathcal{X} \to [0,\infty)$ is such that $$\frac{a Z_{Q,\alpha}}{Q(x)^{\alpha}}  \leq \psi(x)^{-1}$$
	for some $a > 0$, then 
	$$\text{sgn}(\rho) \cdot \mathbb{E}[\psi(X)^{-\rho}] \geq \text{sgn}(\rho) \cdot 2^{\rho(H_{\alpha}(P) + I_{\alpha}(P,Q) + \log a)},$$
	where the expectation is with respect to $P$.
\end{theorem}
\begin{proof}
Since $\psi(x)^{-1} \geq \frac{a Z_{Q,\alpha}}{Q(x)^{\alpha}}$, we have
\begin{align*}
\text{sgn}(\rho) \cdot \mathbb{E}[\psi(X)^{-\rho}] &= \text{sgn}(\rho) \cdot \sum \nolimits_{x \in \mathcal{X}} P(x) \psi(x)^{-\rho} \\
& \geq \text{sgn}(\rho) \cdot a^{\rho} Z_{Q,\alpha}^{\rho} \sum \nolimits_{x \in \mathcal{X}} P(x) Q(x)^{-\alpha\rho} \\
& = \text{sgn}(\rho) \cdot  2^{\rho(H_\alpha(P) + I_{\alpha}(P,Q) + \log a)}.
\end{align*}

\end{proof}

\begin{cor}
\label{cor:upper_bound_unified_mismatch}
Let $\rho > -1$ and $Q$ be a probability distribution on $\mathcal{X}$. Suppose that $\psi: \mathcal{X} \to [0,\infty)$ is such that
$$ \psi(x)^{-1} \leq \frac{c\cdot Z_{Q,\alpha}}{Q(x)^{\alpha}}$$ for some $c > 0$. Then
	\begin{enumerate}
	    \item 
	    $  \frac{1}{\rho} \log \mathbb{E}[\psi(X)^{-\rho}] \leq H_{\alpha}(P) + I_{\alpha}(P,Q) + \log c$\quad for $\rho \neq 0,$
	    \item 
	    $ \mathbb{E} \left[ \log \frac{1}{\psi(X)} \right] \leq H(P) + I(P,Q) + \log c$\quad for $\rho = 0,$
	\end{enumerate}
where the expectation is with respect to $P$.
\end{cor}
\begin{proof}
If $\rho \neq 0$, setting $b = 0$ in Theorem~\ref{thm:bound1_unified_mismatch}, taking log, and dividing throughout by $\rho$, the result follows. If $\rho = 0$, we have $\alpha = 1$, $\psi(x)^{-1} \leq c Q(x)^{-1}$, and 
\begin{align*}
    \mathbb{E} \left[ \log \frac{1}{\psi(X)} \right] &= \sum \nolimits_{x \in \mathcal{X}} P(x) \log (\psi(x)^{-1}) \leq \sum \nolimits_{x \in \mathcal{X}} P(x) \log (c Q(x)^{-1}) \\
        &= \log c + \sum \nolimits_{x \in \mathcal{X}} P(x) \log \frac{P(x)}{Q(x)} - \sum \nolimits_{x \in \mathcal{X}} P(x) \log P(x) \\
        &=  H(P) + I(P,Q) + \log c.
\end{align*}
\end{proof}

\begin{cor}
\label{cor:lower_bound_unified_mismatch}
Let $\rho > -1$ and $Q$ be a probability distribution on $\mathcal{X}$. Suppose $\psi: \mathcal{X} \to [0,\infty)$ is such that 
$$\frac{a Z_{Q,\alpha}}{Q(x)^{\alpha}} \leq \psi(x)^{-1}$$
for some $a > 0$. Then,
\begin{enumerate}
    \item 
    $\frac{1}{\rho} \log \mathbb{E}[\psi(X)^{-\rho}] \geq H_{\alpha}(P) + I_{\alpha}(P,Q) + \log a$\quad for $\rho \neq 0,$
    \item
    $\mathbb{E} \left[ \log \frac{1}{\psi(X)} \right] \geq H(P) + I(P,Q) + \log a$\quad for $\rho = 0,$
	\end{enumerate}
where the expectation is with respect to $P$.
\end{cor}
\begin{proof}
If $\rho \neq 0$, the result follows from Theorem~\ref{thm:bound2_unified_mismatch}, taking log, and dividing throughout by $\rho$. If $\rho = 0$, we have $\alpha = 1$, $\psi(x)^{-1} \geq a Q(x)^{-1}$, and 
\begin{align*}
    \mathbb{E} \left[ \log \frac{1}{\psi(X)} \right] = \sum \nolimits_{x \in \mathcal{X}} P(x) \log (\psi(x)^{-1}) \geq \sum \nolimits_{x \in \mathcal{X}} P(x) \log (a Q(x)^{-1}) =  H(P) + I(P,Q) + \log a.
\end{align*}
\end{proof}

\begin{cor}
\label{cor:sequence_upper_bound_unified_mismatch}
Let $\rho > -1$ and $Q$ be a probability distribution on $\mathcal{X}$. Let $Q_n$ be the $n$-fold product distribution of $Q$ on $\mathcal{X}^n$. Suppose $\psi_n: \mathcal{X}^n \to [0,\infty)$ is such that
	$$\psi_n(x^n)^{-1} \leq \frac{c_n\cdot Z_{Q,\alpha,n}}{Q_n(x^n)^{\alpha}}$$
	for some $c_n > 0$. Then, 
	\begin{enumerate}
	    \item  $ \limsup_{n \to \infty} \frac{1}{n\rho} \log \mathbb{E}[\psi_n(X^n)^{-\rho}] \leq H_{\alpha}(P) + I_{\alpha}(P,Q) + \limsup_{n \to \infty}\frac{1}{n} \log c_n$\quad for $\rho \neq 0,$
	    \item  $ \limsup_{n \to \infty} \frac{1}{n}  \mathbb{E}\left[\log \frac{1}{\psi_n(X^n)} \right] \leq H(P) + I(P,Q) + \limsup_{n \to \infty}\frac{1}{n} \log c_n$\quad for $\rho =0,$
	\end{enumerate}
	where the expectation is with respect to the product distribution $P_n$ of $P$ on $\mathcal{X}^n$.
\end{cor}
\begin{proof}
Using the fact that, for all $\alpha \geq 0$, $H_{\alpha}(P_n) = n H_{\alpha}(P)$ and $I_{\alpha}(P_n,Q_n) = n I_{\alpha}(P,Q)$ after an application of Corollary~\ref{cor:upper_bound_unified_mismatch}, then dividing throughout by $n$ and taking limsup on both sides of the equation, the results follow.
\end{proof}

\begin{cor}
\label{cor:sequence_lower_bound_unified_mismatch}
	Let $\rho > -1$ and $Q$ be a probability distribution on $\mathcal{X}$. Let $Q_n$ be the $n$-fold product distribution of $Q$ on $\mathcal{X}^n$. Suppose $\psi_n: \mathcal{X}^n \to [0,\infty)$ is such that
	$$\frac{a_n Z_{Q,\alpha,n}}{Q_n(x^n)^{\alpha}} \leq \psi_n(x^n)^{-1},$$
	for some $a_n > 0, n \geq 1$. Then, 
	\begin{enumerate}
	    \item $ \liminf_{n \to \infty} \frac{1}{n\rho} \log \mathbb{E}[\psi(X^n)^{-\rho}] \geq H_{\alpha}(P) + I_{\alpha}(P,Q) + \liminf_{n \to \infty}\frac{1}{n} \log a_n$\quad if $\rho \neq 0,$
	    
	    \item $ \liminf_{n \to \infty} \frac{1}{n}  \mathbb{E}\left[\log \frac{1}{\psi(X^n)} \right] \geq H(P) + I(P,Q) + \liminf_{n \to \infty}\frac{1}{n} \log a_n$\quad if $\rho = 0,$
	\end{enumerate}
	where the expectation is with respect to the product distribution $P_n$ of $P$ over the set $\mathcal{X}^n$.
\end{cor}
\begin{proof}
Using the fact that, for all $\alpha > 0$, $H_{\alpha}(P_n) = n H_{\alpha}(P)$ and $I_{\alpha}(P_n,Q_n) = n I_{\alpha}(P,Q)$, after an application of Corollary~\ref{cor:lower_bound_unified_mismatch}, then dividing throughout by $n$ and taking liminf on both sides of the equation, the results follow.
\end{proof}

\section{Proof of the Results by the Unified Approach} \label{sec:proofs_unified}

\subsection{Source Coding Problem} \label{subsec:source_coding}

\medskip
\noindent
\textbf{Proof of Theorem~\ref{thm:source_coding_thm1}}:
Choose $\psi(x) = 2^{-L(C(x))}$, where $L(C(x))$ is the length of codeword $C(x)$ assigned to alphabet $x$. Since $C$ is uniquely decodable, from Proposition~\ref{prop:KraftMcmilan}, we have $\sum_{x \in \mathcal{X}} \psi(x) \leq 1$. Now, an application of Theorem~\ref{thm:lower_bound_logfunc} with $b=1$ tells us that we must have
$\mathbb{E}[L(C(X))] \geq  H(P)$. \hfill $\qed$
\medskip

\noindent
\textbf{Proof of Theorem~\ref{thm:source_coding_thm2}}: Choose $\psi_n(x^n) = 2^{-L_n(x^n)}$, where $L_n(x^n) := L(C(x^n))$ is the length of codeword $C(x^n)$ assigned to sequence $x^n$. Then, an application of Corollary~\ref{cor:sequence_bound_lower} gives us
$$\liminf_{n \to \infty}\mathbb{E}[L_n(C(X))]/n  \geq  H(P).$$
Further, we also have
\begin{align*}
\psi_n(x^n)^{-1} = 2^{L(C(x^n))} = 2^{\lceil - \log P_n(x^n) \rceil}  \leq 2^{- \log P_n(x^n) } = 1/P_n(x^n).   
\end{align*}

Thus, an application of Corollary~\ref{cor:sequence_upper_bound_unified_mismatch} with $c_n = 1$ and $Q = P$ gives us $$\limsup_{n \to \infty}\mathbb{E}[L_n(C(X))]/n  \leq  H(P).$$
\hfill $\qed$

\subsection{Campbell Coding Problem} \label{subsec:campbell_coding}

\medskip
\noindent
\textbf{Proof of Theorem~\ref{thm:campbell_coding_thm1}}: Apply Theorem~\ref{thm:lower_bound_unified} with $\psi(x) = 2^{-L(C(x))}$ and $b=1$. \qed

\medskip
\noindent
\textbf{Proof of Theorem~\ref{thm:campbell_coding_thm2}}: Take $\psi_n(x^n) = 2^{-L(C_n(x^n))}$. Then from \eqref{eqn:optimal_code_length_campbell} we have
\begin{equation*}
	\frac{Z_{P,\alpha,n}}{P_n(x^n)^{\alpha}} \le \psi_n(x^n)^{-1} < 2\cdot \frac{Z_{P,\alpha,n}}{P_n(x^n)^\alpha}\cdot
\end{equation*}
Hence, \eqref{eqn:limit_result_campbell} follows by applying Corollary \ref{cor:sequence_upper_bound_unified_mismatch} with $c_n = 2$ and Corollary \ref{cor:sequence_lower_bound_unified_mismatch} with $a_n = 1$ and $Q_n = P_n$.
\hfill $\qed$

\medskip
\noindent
\textbf{Proof of Theorem~\ref{thm:campbell_coding_mismatch}}: Since length function $K$ satisfies \eqref{eqn:KraftMcmilan}, an application of Theorem~\ref{thm:lower_bound_unified} with $\psi(x) = 2^{-K(x)}$ gives us $\frac{1}{\rho} \log \mathbb{E}[2^{\rho K(X)}]  \geq H_{\alpha}(P)$. Since $K(x) = \lceil \log ({Z_{P,\alpha}}/{P(x)^\alpha)\rceil}$ satisfies \eqref{eqn:KraftMcmilan} and $2^{K(x)} < 2\cdot {Z_{P,\alpha}}/{P(x)^\alpha}$, from Corollary \ref{cor:upper_bound_unified_mismatch}, we know that $\frac{1}{\rho} \log \mathbb{E}[2^{\rho K(X)}]  \leq H_{\alpha}(P) + 1$, that is, the minimum in \eqref{eqn:difference_cumlant} is between $H_\alpha(P)$ and $H_\alpha(P) + 1$. Hence, 
\begin{align}
\frac{1}{\rho} \log \mathbb{E}[2^{\rho L(X)}] - H_{\alpha}(P) -1 \le R_c(P,L,\rho) \le \frac{1}{\rho}  \log  \mathbb{E}[2^{\rho L(X)}] - H_\alpha(P).
    \label{eqn:upper_lower_bound_difference_cumulant}
\end{align}
Let us now define a probability distribution $Q_L$ as 
\begin{equation*}
	Q_L(x) = \frac{2^{-(1+\rho)L(x)}}{\sum_{x'}2^{-(1+\rho)L(x')}}\cdot
	\end{equation*}
Then,
\begin{equation*}
2^{L(x)} = Q_L(x)^{-\alpha} Z_{Q_L,\alpha}\cdot \frac{1}{\sum_{x'}2^{-L(x')}} = Q_L(x)^{-\alpha} Z_{Q_L,\alpha}\cdot \frac{1}{\eta},
\end{equation*}
where $\eta = \sum_{x'}2^{-L(x')}$. Applying Corollaries \ref{cor:upper_bound_unified_mismatch} and \ref{cor:lower_bound_unified_mismatch} with $\psi(x) = 2^{-L(x)}$, we get
\begin{equation}
\label{eqn:upper_lower_bound_mismatch_campbell}
\frac{1}{\rho} \log \mathbb{E}[2^{\rho L(X)}] = H_\alpha(P) + I_\alpha(P,Q_L) - \log \eta.
\end{equation}
Combining (\ref{eqn:upper_lower_bound_difference_cumulant}) and (\ref{eqn:upper_lower_bound_mismatch_campbell}), we get the desired result. \qed


\subsection{Arikan's Guessing Problem} \label{subsec:massey_gessing}
\medskip
\noindent
\textbf{Proof of Theorem~\ref{thm:massey_guessing_thm1}}: Let $G$ be any guessing function. Choose $\psi(x) = 1/G(x)$. Then, we have $\sum_{x \in \mathcal{X}} \psi(x) = \sum_{x \in \mathcal{X}} 1/G(x) = \sum^M_{i=1} 1/i = h_M$ (due to the bijective nature of function $G$). Now, an application of Theorem~\ref{thm:lower_bound_unified} with $b = h_M$, gives us 
$$\frac{1}{\rho} \log \mathbb{E}[G(X)^{\rho}] = \frac{1}{\rho} \log \mathbb{E}[\psi(X)^{-\rho}] \geq H_{\alpha}(P) - \log h_M.$$
\hfill $\qed$

\medskip
\noindent
\textbf{Proof of Theorem~\ref{thm:massey_guessing_thm2}}: Let us rearrange the probabilities $\{P(x) , x \in \mathcal{X}\}$ in non-increasing order, say
$$p_1 \geq p_2 \geq \cdots \geq p_M.$$
Then, the optimal guessing function $G^{*}$ is given by $G^{*}(x) = i$ if $P(x) = p_i$. Let us index the elements in the set $\mathcal{X}$ as $\{x_1, x_2, \ldots, x_M \}$, according to the decreasing order of their probabilities. Then, for $i\in \{1,\dots, M\}$, we have
$$\frac{Z_{P,\alpha}}{P(x_i)^{\alpha}} = \frac{\sum^{M}_{j=1} p^{\alpha}_j}{p^{\alpha}_i} \geq i =  G^{*}(x_i).$$
That is, $G^{*}(x) \leq \frac{Z_{P,\alpha}}{P(x)^{\alpha}} \, \,  \text{fpr } x \in \mathcal{X}$. Now, an application of Corollary~\ref{cor:upper_bound_unified_mismatch} with $Q = P$, $\psi(x) = 1/G^{*}(x)$, and $c = 1$, we get 
\begin{align*}\frac{1}{\rho} \log \mathbb{E}[G^{*}(X)^{\rho}] =  \frac{1}{\rho} \log \mathbb{E}[\psi(X)^{-\rho}] \leq H_{\alpha}(P) + I_{\alpha}(P,Q) + \log 1 = H_{\alpha}(P).
\end{align*}
\hfill $\qed$

\medskip
\noindent
\textbf{Proof of Theorem~\ref{thm:massey_guessing_thm3}}: Let $G_n^{*}$ be the optimal guessing function from the set of $n$-length sequences $\mathcal{X}^{n}$ to the set of natural numbers $\{1,2,\ldots, M^n\}$. An application of Corollary~\ref{cor:sequence_bound_lower} with $\psi(x^n) = 1/G_n^{*}(x_n)$ and $b_n = h_{M^n}$ gives us
$$\liminf_{n \to \infty} \frac{1}{n\rho} \log \mathbb{E}[G_n^*(X^n)^{\rho}] \geq H_{\alpha}(P) - \limsup_{n \to \infty} \frac{\log h_{M^n}}{n}.$$
We note that 
$$\limsup_{n \to \infty} \frac{\log h_{M^n}}{n} \leq  \limsup_{n \to \infty} \frac{\log (2n \log M)}{n} = 0.$$
Hence,
\begin{align} \label{eq:massey_thm3_eq1}
\liminf_{n \to \infty} \frac{1}{n\rho} \log \mathbb{E}[G_n^*(X^n)^{\rho}] \geq H_{\alpha}(P)    
\end{align}

From the proof of the previous theorem, we know that $G^{*}(x^n) \leq \frac{Z_{P,\alpha,n}}{P_n(x^n)^{\alpha}} \, \,  \forall x^n \in \mathcal{X}^n$. Now, an application of Corollary~\ref{cor:sequence_upper_bound_unified_mismatch} with $Q_n = P_n$, $\psi_n(x) = 1/G_n^{*}(x)$, and $c_n = 1$, gives us
\begin{align} \label{eq:massey_thm3_eq2}
\limsup_{n \to \infty} \frac{1}{n\rho} \log \mathbb{E}[G_n^{*}(X^n)^{\rho}] \leq H_{\alpha}(P).
\end{align}
Combining \eqref{eq:massey_thm3_eq1} and \eqref{eq:massey_thm3_eq2}, we get the desired result. \hfill $\qed$ \\

\medskip
\noindent
\textbf{Proof of Theorem~\ref{thm:guessing_mismatch_thm1}}: Let $G$ be a guessing function. Define a probability distribution $Q_G = \{Q_G(x),~x \in \mathcal{X}\}$ as
$$Q_G(x) = \frac{G(x)^{-(1+\rho)}}{\sum_{x' \in \mathcal{X}} G(x')^{-(1+\rho)}}\cdot$$
Then, we have
$$\frac{Z_{Q_G,\alpha}}{Q_G(x)^{\alpha}} = G(x) \sum \nolimits_{x' \in \mathcal{X}} \frac{1}{G(x')} = G(x)\cdot h_M\cdot$$
That is, $G(x) = \frac{Z_{Q_G,\alpha}}{h_M\cdot Q_G(x)^{\alpha}}$. Now, an application of Corollaries~\ref{cor:upper_bound_unified_mismatch} and \ref{cor:lower_bound_unified_mismatch} with $\psi(x) = 1/G(x)$, $c = 1/h_M$ and $a = 1/h_M$  yields the desired result.
\hfill \qed \\

\medskip
\noindent
\textbf{Proof of Theorem~\ref{thm:guessing_mismatch_thm2}}:  Let us rearrange the probabilities $\{Q(x) , x \in \mathcal{X}\}$ in non-increasing order, say
$$q_1 \geq q_2 \geq \cdots \geq q_M.$$

Let $G_Q$ be the guessing strategy that guesses according to the decreasing order of $Q$-probabilities, that is, $G_Q(x) = i$ if $Q(x) = q_i$. Then, as in the proof of Theorem~\ref{thm:massey_guessing_thm2}, we have $G_Q(x) \leq {Z_{Q,\alpha}}/{Q(x)^{\alpha}} \quad  \text{for } x \in \mathcal{X}$. Hence, an application of Corollary~\ref{cor:upper_bound_unified_mismatch} with $\psi(x) = 1/G_Q(x)$ and $c = 1$ proves the result. \hfill $\qed$

\begin{remark}
We remark here that the above proof also covers the case when $\rho \in (-1,0)$, whereas the original proof of Arikan and Sundaresan dealt only the case $\rho > 0$.
\end{remark}

\medskip
\noindent
\textbf{Proof of Theorem~\ref{thm:guessing_mismatch_thm3}}: Combining Theorems~\ref{thm:massey_guessing_thm2} and \ref{thm:guessing_mismatch_thm1} we get $R_{g}(P,G,\rho) \geq I_\alpha(P,Q_G) - \log h_M$. From Theorem~\ref{thm:massey_guessing_thm1}, we know that $\frac{1}{\rho} \log \mathbb{E}[G_P(X)^{\rho}] \geq H_{\alpha}(P) - \log h_M$. This along with Theorem~\ref{thm:guessing_mismatch_thm1} tell us that  $R_{g}(P,G,\rho) \leq I_{\alpha}(P,Q_G)$. 
\hfill \qed

\subsection{Memoryless Guessing} \label{subsec:memoryless_guessing}

\noindent
\textbf{Proof of Theorem~\ref{thm:memoryless_guessing_thm1}}: Consider
	\begin{align*}
	\mathbb{E}[V_{\rho}(X,\widehat{X}^{\infty}_1)]=\sum \nolimits_{m=1}^{\infty}
	P\{G(X,\widehat{X}^{\infty}_{1})=m\}V_{\rho}(X,\widehat{X}^{\infty}_1).
	\end{align*}
	Since
	\begin{align*}
	P\{G(X,\widehat{X}^{\infty}_{1})=m\} =\sum \nolimits_{x\in\mathcal{X}} P(x)
	\widehat{P}(x)(1-\widehat{P}(x))^{m-1},
	\end{align*}
	we obtain
    \begin{align}
	\mathbb{E}[V_{\rho}(X,\widehat{X}^{\infty}_1)]  &=  \sum_{m=1}^{\infty}\sum_{x\in\mathcal{X}} P(x) \widehat{P}(x)(1-\widehat{P}(x))^{m-1}  \frac{1}{\rho!} \prod \nolimits_{l=0}^{\rho -1}(m+l) \nonumber \\
	&=  \sum \nolimits_x P(x)\widehat{P}(x)^{-\rho} = \mathbb{E}[\widehat{P}(X)^{-\rho}], \label{eqn:factorial_moment}
	\end{align}
	where (\ref{eqn:factorial_moment}) is by \cite[p.~2]{Huieihel}. Now, the result follows from Theorem~\ref{thm:lower_bound_unified} with $\psi(x) = \widehat{P}(x) $ and $b = 1$. Since $\widehat{P}$ is a probability distribution, we have $\sum_{x \in \mathcal{X}} \psi(x) = \sum_{x \in \mathcal{X}} \widehat{P}(x) = 1$. Hence, by Theorem~\ref{thm:lower_bound_unified}, $\frac{1}{\rho} \log \mathbb{E}[V^{*}_{\rho}(X,\widehat{X}^{\infty}_1)] = H_{\alpha}(P)$,
	attained by $\widehat{P}^{*}(x) = P(x)^{\alpha} / Z_{P,\alpha}$. 
\hfill \qed
\medskip

\noindent
\textbf{Proof of Theorem~\ref{thm:memoryless_mismatch}}: The result follows easily by taking $\psi = \widehat{Q}^*$, $c=1, a=1$ in Corollaries \ref{cor:upper_bound_unified_mismatch} and \ref{cor:lower_bound_unified_mismatch}. \hfill \qed

\subsection{Task Partitioning Problem} \label{subsec:tasks_partition}

In this section, we establish all the results concerning the tasks partitioning problem discussed in Section~\ref{sec:probelm_statements} using the unified approach. In the following, Propositions \ref{prop:kraft_mcmillan_tasks1} and \ref{prop:kraft_mcmillan_tasks2} are the analogue of Kraft-McMillan theorem for partition functions.

\begin{proposition}\cite[Prop.~III 1]{Bunte}:
\label{prop:kraft_mcmillan_tasks1}
 Let $A$ be the partition function associated with a partition of $\mathcal{X}$ of size $N$. Then,
\begin{equation}
\label{eqn:partition_identity}
\sum \nolimits_{x\in\mathcal{X}}\frac{1}{A(x)} = N.
\end{equation}
\end{proposition}
\medskip

\begin{proposition} \label{prop:kraft_mcmillan_tasks2}
{\rm \cite[Prop.~III.2]{Bunte}} If $\lambda : \mathcal{X} \to \mathbb{N} \cup \{+\infty \}$ and $\sum_{x \in \mathcal{X}} 1/\lambda(x) = \mu$, then there exists a partition of $\mathcal{X}$ of size at most
$$\min_{\alpha > 1} \lfloor \alpha \mu + \log_{\alpha} |\mathcal{X}| + 2 \rfloor $$
such that its partition function $A$ satisfies $A(x) \leq \min \{\lambda(x), |\mathcal{X}| \}$ for $x \in \mathcal{X}$.
\end{proposition}

\begin{cor} \label{cor:bunte_cor1}
Let $N$ be such that $N > \log|\mathcal{X}| + 2$. If $\lambda : \mathcal{X} \to \mathbb{N} \cup \{+\infty \}$ is such that $\sum_{x \in \mathcal{X}} 1/\lambda(x) = (N - \log|\mathcal{X}| - 2)/2$, then there exists a partition of $\mathcal{X}$ of size at most $N$ with partition function $A$ such that $A(x) \leq \min \{\lambda(x), |\mathcal{X}| \}$ for $x \in \mathcal{X}$.
\end{cor}
\begin{proof}
Let $|\mathcal{A}|$ denote the size of partition $\mathcal{A}$. Applying Proposition \ref{prop:kraft_mcmillan_tasks2} with $\mu = (N - \log|\mathcal{X}| - 2)/2$, we see that there exists a partition $\mathcal{A}$ of $\mathcal{X}$ with partition function $A$ such that
$$
|\mathcal{A}| \leq  \min_{\alpha > 1} \lfloor \alpha \mu + \log_{\alpha} |\mathcal{X}| + 2 \rfloor \leq \lfloor 2 \mu + \log |\mathcal{X}| + 2 \rfloor \leq N
$$
and $A(x) \leq \min \{\lambda(x), |\mathcal{X}| \} \text{ for } x \in \mathcal{X}$.
\end{proof}
\medskip

\noindent
\textbf{Proof of Theorem~\ref{thm:task_partition_thm1}(a)}: Follows from an application of Theorem~\ref{thm:lower_bound_unified} with $\psi(x) = 1/A(x)$ and $b = N$, and by Proposition \ref{prop:kraft_mcmillan_tasks1}.
\hfill \qed
\medskip

\noindent
\textbf{Proof of Theorem~\ref{thm:task_partition_thm1}(b)}:
Let
\begin{align*}
\lambda(x) = \lceil \beta\cdot P(x)^{-\alpha}\rceil,\quad\text{ with }\beta = \frac{2\cdot Z_{P,\alpha}}{N-\log |\mathcal{X}|-2}\cdot
\end{align*}
Then,
\begin{equation*}
	 \mu = \sum_{x \in \mathcal{X}} \frac{1}{\lambda(x)} \le \sum_x \frac{1}{\beta\cdot P(x)^{-\alpha}} \le \frac{N-\log |\mathcal{X}|-2}{2}\cdot
\end{equation*}

Now an application of Corollary~\ref{cor:bunte_cor1} tells us that there exists a partition of $\mathcal{X}$ of size at most $N$ with partition function $A$ such that $A(x) \leq \min \{\lambda(x), |\mathcal{X}| \}$ for $x \in \mathcal{X}$. Further, for any $\rho \neq 0$, we also have
\begin{align*}
A(x) \leq \lambda(x) =  \left( \lceil \beta\cdot P(x)^{-\alpha}\rceil^{|\rho|} \right)^{\frac{1}{|\rho|}} \leq  \left( 1+ 2^{|\rho|}(\beta\cdot P(x)^{-\alpha})^{|\rho|} \right)^{\frac{1}{|\rho|}} = \left(1 + \left( \frac{Z_{P,\alpha}}{\tilde{N}  P(x)^{\alpha}} \right)^{|\rho|} \right)^{\frac{1}{|\rho|}},
\end{align*}
where $\tilde{N} = \frac{N-\log |\mathcal{X}|-2}{4}$. Now, an application of Theorem~\ref{thm:bound1_unified_mismatch} with $\psi(x) = 1/A(x)$, $b=1$, $c = 1/\tilde{N}$, and $Q=P$ gives us 
$$\mathbb{E}[A(X)^{\rho}] \leq 1 + 2^{\rho(H_{\alpha}(P)- \log \tilde{N})} \qquad \textrm{if } \rho > 0$$
and if $\rho < 0$, we have
\begin{align*}
\mathbb{E}[A(X)^{\rho}]  \geq 2^{\rho(H_{\alpha}(P) - \log \tilde{N} - \rho^{-1} \log (1+ \tilde{N}^{-\rho}) )} = \left(\frac{\tilde{N}^{-\rho}}{1+\tilde{N}^{-\rho}} \right) 2^{\rho H_{\alpha}(P)} \geq \frac{1}{2} \cdot 2^{\rho H_{\alpha}(P)}.
\end{align*}
Also, if $\rho <0$, since $A(X) \ge 1$, we have $A(X)^{\rho} \le 1$. Hence $\mathbb{E}[A(X)^{\rho}] \le 1$. This completes the proof.
\hfill \qed


\medskip
\noindent
\textbf{Proof of Theorem~\ref{thm:task_partition_thm2}}:
\medskip

\noindent
\textbf{Part (i)} Proceeding as in the proof of Theorem~\ref{thm:task_partition_thm1}(b), we can show that for any $n \geq 1$, there exists a partition $\mathcal{A}_n$ of $\mathcal{X}^n$ of size $N^n$ such that the associated partition function $A_n(\cdot)$ satisfies 
$$\mathbb{E}[A_n(X^n)^{\rho}] \leq 1 + 2^{\rho(H_{\alpha}(P_n)- \log \tilde{N}_n)},$$
where $\tilde{N}_n = (N^n - n \log |\mathcal{X}|  - 2)/2$. We note that $P_n$ is the $n$-fold product distribution of $P$ on $\mathcal{X}^n$. Thus, we have $H_{\alpha}(P_n) = n H_{\alpha}(P)$. Now, if $\log N > H_{\alpha}(P)$, there exists $\delta > 0$ such that $\log N > H_{\alpha}(P) + \delta$. Since $\lim_{n \to \infty} \frac{\log \tilde{N}_n}{n} = \log N$, there exists a positive integer $n_{\delta}$ such that $\log N - \delta < \frac{\log \tilde{N}_n}{n}$ for $n \geq n_{\delta}$. Thus,  we have
$$\mathbb{E}[A_n(X^n)^{\rho}] < 1 + 2^{n \rho(H_{\alpha}(P)- \log N + \delta)} \qquad \text{for } n \geq n_{\delta}.$$

Since $\log N > H_{\alpha}(P) + \delta$, we have $\rho(H_{\alpha}(P)- \log N + \delta) < 0$. Consequently, $\limsup_{n \to \infty} \mathbb{E}[A_n(X^n)^{\rho}] \leq 1 $. Since $A_n(x^n) \geq 1$ for all $x^n \in \mathcal{X}^n$, we have $\liminf_{n \to \infty} \mathbb{E}[A_n(X^n)^{\rho}] \geq 1 $. Thus, we conclude that $\lim_{n \to \infty} \mathbb{E}[A_n(X^n)^{\rho}] = 1$.
\medskip

\noindent
\textbf{Part (ii)} For each $n\geq 1$, let $\mathcal{A}_n$ be a partition of $\mathcal{X}^n$ of size $N^n$ with partition function $A_n$. An application of Theorem~\ref{thm:lower_bound_unified} with $\psi(\cdot) = 1/A_n(\cdot)$ and $b=N^n$ gives us 
$$\mathbb{E}[A_n(X^n)^{\rho}] \geq 2^{\rho(H_{\alpha}(P_n)- \log N^n)} =  2^{n\rho(H_{\alpha}(P)- \log N)}.$$
Now, since $\log N < H_{\alpha}(P)$ and $\rho > 0$, we have $\lim_{n \to \infty} \mathbb{E}[A_n(X^n)^{\rho}] = \infty$.
\hfill \qed 

\medskip
\noindent
\textbf{Proof of Theorem~\ref{thm:task_partition_mismatch_lb}}: Define a probability distribution $Q_A = \{Q_A(x),~x \in \mathcal{X}\}$ as
$$Q_A(x) = \frac{A(x)^{-(1+\rho)}}{\sum_{x' \in \mathcal{X}} A(x')^{-(1+\rho)}}\cdot$$
Then, we have
$$\frac{Z_{Q_A,\alpha}}{Q_A(x)^{\alpha}} = A(x) \sum \nolimits_{x' \in \mathcal{X}} \frac{1}{A(x')} = A(x)\cdot N,$$
where the last equality follows due to Proposition~\ref{prop:kraft_mcmillan_tasks1}. Rearranging the terms, we have $A(x)=\frac{Z_{Q_A,\alpha}}{N \cdot Q_A(x)^{\alpha}}$. Now, an application of Corollaries~\ref{cor:upper_bound_unified_mismatch} and \ref{cor:lower_bound_unified_mismatch} with $\psi(x) = 1/A(x)$, $c = 1/N$ and $a = 1/N$  yields the desired result.
\hfill \qed \\

\noindent
\textbf{Proof of Theorem~\ref{thm:task_partition_mismatch_ub}}: The proof is identical to that of Theorem~\ref{thm:task_partition_thm1}(b), with the exception that we need to invoke Theorem~\ref{thm:bound1_unified_mismatch} with $\psi(x) = 1/A_{Q}(x)$, $b=1$ and $c = 1/\tilde{N}$. \hfill \qed

\section{Summary and Concluding Remarks} \label{sec:summary}
The purpose of this paper was to establish a close relationship between the problems on source coding, guessing, and task partitioning by unifying all the proofs of the results known in these problems. Indeed, we could establish the lower bound of these problems using an optimization problem and apply Campbell's idea to achieve the lower bound in an asymptotic way. Our approach also enabled us to refine some of the known results as well as obtain some new results for the mismatched version of these problems. We were able to formulate a more general mathematical problem where the above mentioned problems are particular instances. 

We feel that this  generalization, in  addition  to help solving new  problems  that  fall  in  this framework, would also provide new insights. For example, Arikan's guessing problem can be interpreted as a particular case of a more general task partitioning problem, where tasks assigned to a particular key also have to be ordered.

The presented unified approach can be explored further in many ways. This includes, (i) Extension to general state-space: It would be interesting to see if the studied problems can be formulated and solved, for example, for countably infinite or continuous support sets. (ii) Applications: Arikan showed an application of the guessing problem in a sequential decoding problem \cite{Arikan}. Humblet showed that
cumulants of code-lengths in Campbell's problem has applications in
minimizing the probability of buffer overflow in source coding problems \cite{Humblet}. We would like to see if other potential applications emerge from this unified study.



\bibliographystyle{IEEEtran}
\bibliography{Unification}

\end{document}